\documentclass[12pt]{amsart}
\usepackage{mathrsfs}      
\usepackage{txfonts}
\usepackage{amssymb}
\usepackage{eucal}
\usepackage{graphicx}
\usepackage{amsmath}
\usepackage{amscd}
\usepackage{extarrows}
\usepackage[all]{xy}           
\usepackage[active]{srcltx} 
\usepackage{multirow}
\usepackage{amsfonts,latexsym}
\usepackage{xspace}
\usepackage{epsfig}
\usepackage{float}
\usepackage{color}
\usepackage{fancybox}
\usepackage{colordvi}
\usepackage{multicol}
\usepackage{colordvi}
\usepackage{tikz}
\usepackage[colorlinks,final,backref=page,hyperindex,hypertex]{hyperref}

\newtheorem{theorem}{Theorem}[section]
\newtheorem{prop}[theorem]{Proposition}

\newtheorem{lemma}[theorem]{Lemma}

\newtheorem{prop-def}{Proposition-Definition}[section]
\newtheorem{coro-def}{Corollary-Definition}[section]

\newcommand{\nc}{\newcommand}
\nc{\tred}[1]{\textcolor{red}{#1}}
\nc{\tblue}[1]{\textcolor{blue}{#1}}
\nc{\tgreen}[1]{\textcolor{green}{#1}}
\nc{\tpurple}[1]{\textcolor{purple}{#1}}
\nc{\btred}[1]{\textcolor{red}{\bf #1}}
\nc{\btblue}[1]{\textcolor{blue}{\bf #1}}
\nc{\btgreen}[1]{\textcolor{green}{\bf #1}}
\nc{\btpurple}[1]{\textcolor{purple}{\bf #1}}

\renewcommand{\frak}{\mathfrak}

\newcommand{\efootnote}[1]{}

\renewcommand{\textbf}[1]{}

\newcommand{\delete}[1]{}
\nc{\dfootnote}[1]{{}}          
\nc{\ffootnote}[1]{\dfootnote{#1}}
\nc{\mfootnote}[1]{\footnote{#1}} 
\nc{\ofootnote}[1]{\footnote{\tiny Older version: #1}}

\nc{\mlabel}[1]{\label{#1}}  
\nc{\mcite}[1]{\cite{#1}}  
\nc{\mref}[1]{\ref{#1}}  

\delete{
\nc{\mlabel}[1]{\label{#1}  
{\hfill \hspace{1cm}{\bf{{\ }\hfill(#1)}}}}
\nc{\mcite}[1]{\cite{#1}{{\bf{{\ }(#1)}}}}  
\nc{\mref}[1]{\ref{#1}{{\bf{{\ }(#1)}}}}  
}
\nc{\mbibitem}[1]{\bibitem[\bf #1]{#1}} 
\nc{\mkeep}[1]{\marginpar{{\bf #1}}} 

\nc{\opa}{\ast} \nc{\opb}{\odot} \nc{\op}{\bullet} \nc{\pa}{\frakL}
\nc{\arr}{\rightarrow} \nc{\lu}[1]{(#1)} \nc{\mult}{\mrm{mult}}
\nc{\diff}{\mathfrak{Diff}}
\nc{\opc}{\sharp}\nc{\opd}{\natural}
\nc{\ope}{\circ}
\nc{\AW}{{\bf AW}}

\nc{\bin}[2]{ (_{\stackrel{\scs{#1}}{\scs{#2}}})}  
\nc{\binc}[2]{ \left (\!\! \begin{array}{c} \scs{#1}\\
    \scs{#2} \end{array}\!\! \right )}  
\nc{\bincc}[2]{  \left ( {\scs{#1} \atop
    \vspace{-1cm}\scs{#2}} \right )}  
\nc{\bs}{\bar{S}} \nc{\cosum}{\sqsubset} \nc{\la}{\longrightarrow}
\nc{\rar}{\rightarrow} \nc{\dar}{\downarrow} \nc{\dprod}{**}
\nc{\dap}[1]{\downarrow \rlap{$\scriptstyle{#1}$}}
\nc{\md}{\mathrm{dth}} \nc{\uap}[1]{\uparrow
\rlap{$\scriptstyle{#1}$}} \nc{\defeq}{\stackrel{\rm def}{=}}
\nc{\disp}[1]{\displaystyle{#1}} \nc{\dotcup}{\
\displaystyle{\bigcup^\bullet}\ } \nc{\gzeta}{\bar{\zeta}}
\nc{\hcm}{\ \hat{,}\ } \nc{\hts}{\hat{\otimes}}
\nc{\barot}{{\otimes}} \nc{\free}[1]{\bar{#1}}
\nc{\uni}[1]{\tilde{#1}} \nc{\hcirc}{\hat{\circ}} \nc{\lleft}{[}
\nc{\lright}{]} \nc{\lc}{\lfloor} \nc{\rc}{\rfloor}
\nc{\curlyl}{\left \{ \begin{array}{c} {} \\ {} \end{array}
    \right .  \!\!\!\!\!\!\!}
\nc{\curlyr}{ \!\!\!\!\!\!\!
    \left . \begin{array}{c} {} \\ {} \end{array}
    \right \} }
\nc{\longmid}{\left | \begin{array}{c} {} \\ {} \end{array}
    \right . \!\!\!\!\!\!\!}
\nc{\onetree}{\bullet} \nc{\ora}[1]{\stackrel{#1}{\rar}}
\nc{\ola}[1]{\stackrel{#1}{\la}}
\nc{\ot}{\otimes} \nc{\mot}{{{\boxtimes\,}}}
\nc{\otm}{\overline{\boxtimes}} \nc{\sprod}{\bullet}
\nc{\scs}[1]{\scriptstyle{#1}} \nc{\mrm}[1]{{\rm #1}}
\nc{\margin}[1]{\marginpar{\rm #1}}   
\nc{\dirlim}{\displaystyle{\lim_{\longrightarrow}}\,}
\nc{\invlim}{\displaystyle{\lim_{\longleftarrow}}\,}
\nc{\mvp}{\vspace{0.3cm}} \nc{\tk}{^{(k)}} \nc{\tp}{^\prime}
\nc{\ttp}{^{\prime\prime}} \nc{\svp}{\vspace{2cm}}
\nc{\vp}{\vspace{8cm}} \nc{\proofbegin}{\noindent{\bf Proof: }}
\nc{\proofend}{$\blacksquare$ \vspace{0.3cm}}
\nc{\modg}[1]{\!<\!\!{#1}\!\!>}
\nc{\intg}[1]{F_C(#1)} \nc{\lmodg}{\!
<\!\!} \nc{\rmodg}{\!\!>\!}
\nc{\cpi}{\widehat{\Pi}}
\nc{\sha}{{\mbox{\cyr X}}}  
\nc{\shap}{{\mbox{\cyrs X}}} 
\nc{\shpr}{\diamond}    
\nc{\shp}{\ast} \nc{\shplus}{\shpr^+}
\nc{\shprc}{\shpr_c}    
\nc{\msh}{\ast} \nc{\zprod}{m_0} \nc{\oprod}{m_1}
\nc{\vep}{\varepsilon} \nc{\labs}{\mid\!} \nc{\rabs}{\!\mid}
\nc{\speciall}{\mathrm{sl(2,\mathbb{C})}}

\nc{\mmbox}[1]{\mbox{\ #1\ }} \nc{\fp}{\mrm{FP}}
\nc{\rchar}{\mrm{char}} \nc{\End}{\mrm{End}} \nc{\Fil}{\mrm{Fil}}
\nc{\Mor}{Mor\xspace} \nc{\gmzvs}{gMZV\xspace}
\nc{\gmzv}{gMZV\xspace} \nc{\mzv}{MZV\xspace}
\nc{\mzvs}{MZVs\xspace} \nc{\Hom}{\mrm{Hom}} \nc{\id}{\mrm{id}}
\nc{\im}{\mrm{im}} \nc{\incl}{\mrm{incl}} \nc{\map}{\mrm{Map}}
\nc{\mchar}{\rm char} \nc{\nz}{\rm NZ} \nc{\supp}{\mathrm Supp}

\nc{\Alg}{\mathbf{Alg}} \nc{\Bax}{\mathbf{Bax}} \nc{\bff}{\mathbf f}
\nc{\bfk}{{\bf k}} \nc{\bfone}{{\bf 1}} \nc{\bfx}{\mathbf x}
\nc{\bfy}{\mathbf y}
\nc{\base}[1]{\bfone^{\otimes ({#1}+1)}} 
\nc{\Cat}{\mathbf{Cat}}

\nc{\detail}{\marginpar{\bf More detail}
    \noindent{\bf Need more detail!}
    \svp}
\nc{\Int}{\mathbf{Int}} \nc{\Mon}{\mathbf{Mon}}
\nc{\rbtm}{{shuffle }} \nc{\rbto}{{Rota-Baxter }}
\nc{\remarks}{\noindent{\bf Remarks: }} \nc{\Rings}{\mathbf{Rings}}
\nc{\Sets}{\mathbf{Sets}} \nc{\wtot}{\widetilde{\odot}}
\nc{\wast}{\widetilde{\ast}} \nc{\bodot}{\bar{\odot}}
\nc{\bast}{\bar{\ast}} \nc{\hodot}[1]{\odot^{#1}}
\nc{\hast}[1]{\ast^{#1}} \nc{\mal}{\mathcal{O}}
\nc{\tet}{\tilde{\ast}} \nc{\teot}{\tilde{\odot}}
\nc{\oex}{\overline{x}} \nc{\oey}{\overline{y}}
\nc{\oez}{\overline{z}} \nc{\oef}{\overline{f}}
\nc{\oea}{\overline{a}} \nc{\oeb}{\overline{b}}
\nc{\weast}[1]{\widetilde{\ast}^{#1}}
\nc{\weodot}[1]{\widetilde{\odot}^{#1}} \nc{\hstar}[1]{\star^{#1}}
\nc{\lae}{\langle} \nc{\rae}{\rangle}
\nc{\lf}{\lfloor}\nc{\rf}{\rfloor}
\newcommand{\threech}[3]{\par\begin{tabular}{*{3}{@{}p{5.2cm}}}(1)~#1
& (2)~#2 & (3)~#3\end{tabular}}

\newcommand{\twoch}[2]{\par\begin{tabular}{*{2}{@{}p{7cm}}}(4)~#1
& (5)~#2 \end{tabular}}

\newcommand{\onechss}[1]{\par\begin{tabular}{*{1}{@{}p{14cm}}}(6)~#1 \end{tabular}}

\newcommand{\onech}[1]{\par\begin{tabular}{*{1}{@{}p{14cm}}}(7)~#1 \end{tabular}}

\newcommand{\onechs}[1]{\par\begin{tabular}{*{1}{@{}p{14cm}}}(8)~#1 \end{tabular}}

\newcommand{\twochs}[2]{\par\begin{tabular}{*{2}{@{}p{7cm}}}(9)~#1
& (10)~#2 \end{tabular}}

\nc{\CC}{\mathbb{C}}   \nc{\RR}{\mathbf{RR}}

\nc{\cala}{{\mathcal A}} \nc{\calb}{{\mathcal B}}
\nc{\calc}{{\mathcal C}}
\nc{\cald}{{\mathcal D}} \nc{\cale}{{\mathcal E}}
\nc{\calf}{{\mathcal F}} \nc{\calg}{{\mathcal G}}
\nc{\calh}{{\mathcal H}} \nc{\cali}{{\mathcal I}}
\nc{\call}{{\mathcal L}} \nc{\calm}{{\mathcal M}}
\nc{\caln}{{\mathcal N}} \nc{\calo}{{\mathcal O}}
\nc{\calp}{{\mathcal P}} \nc{\calr}{{\mathcal R}}
\nc{\cals}{{\mathcal S}} \nc{\calt}{{\mathcal T}}
\nc{\calu}{{\mathcal U}} \nc{\calw}{{\mathcal W}} \nc{\calk}{{\mathcal K}}
\nc{\calx}{{\mathcal X}} \nc{\CA}{\mathcal{A}}

\nc{\fraka}{{\mathfrak a}} \nc{\frakA}{{\mathfrak A}}
\nc{\frakb}{{\mathfrak b}} \nc{\frakB}{{\mathfrak B}}
\nc{\frakD}{{\mathfrak D}} \nc{\frakg}{{\mathfrak g}}
\nc{\frakH}{{\mathfrak H}} \nc{\frakL}{{\mathfrak L}}
\nc{\frakM}{{\mathfrak M}} \nc{\bfrakM}{\overline{\frakM}}
\nc{\frakm}{{\mathfrak m}} \nc{\frakP}{{\mathfrak P}}
\nc{\frakN}{{\mathfrak N}} \nc{\frakp}{{\mathfrak p}}
\nc{\frakS}{{\mathfrak S}}

\font\cyr=wncyr10 \font\cyrs=wncyr7
\nc{\li}[1]{\textcolor{red}{Li:#1}}
\nc{\jun}[1]{\textcolor{blue}{Jun: #1}}
\nc{\cm}[1]{\textcolor{purple}{CM: #1}}

\begin{document}

\title{Rota-Baxter operators on $\speciall$ and solutions of the classical Yang-Baxter equation}

\author{Jun Pei}
\address{Department of Mathematics, Lanzhou University, Lanzhou, Gansu 730000, China}
         \email{peitsun@163.com}

\author{Chengming Bai}
\address{Chern Institute of Mathematics \& LPMC, Nankai University, Tianjin 300071, China}
         \email{baicm@nankai.edu.cn}

\author{Li Guo}
\address{Department of Mathematics and Computer Science,
         Rutgers University,
         Newark, NJ 07102 }
\email{liguo@rutgers.edu}

\date{\today}

\begin{abstract}
We explicitly determine all Rota-Baxter operators (of weight zero) on $\speciall$ under the Cartan-Weyl basis. For the skew-symmetric operators, we give the corresponding skew-symmetric solutions of the classical Yang-Baxter equation in $\speciall$, confirming the related study by Semenov-Tian-Shansky. In general, these Rota-Baxter operators give a family of solutions of the classical Yang-Baxter equation in the 6-dimensional Lie algebra $\speciall \ltimes_{{\rm ad}^{\ast}} \speciall^{\ast}$. They also give rise to 3-dimensional pre-Lie algebras which in turn yield solutions of the classical Yang-Baxter equation in other 6-dimensional Lie algebras.
\end{abstract}

\subjclass[2010]{16T25, 81R15}

\keywords{Rota-Baxter operator, classical Yang-Baxter equation, pre-Lie algebra}

\maketitle

\tableofcontents

\setcounter{section}{0}


\section{Introduction}

A {\bf Rota-Baxter operator (of weight zero) on an associative algebra $A$}\footnote{More generally, for any given scalar $\lambda$, a Rota-Baxter operator of weight $\lambda$ on $A$ is a linear map $P:A\rightarrow A$ satisfying
$P(x)P(y)=P(P(x)y+xP(y)+\lambda x y), \forall x,y\in A.$ We will focus on the weight zero case in this paper, both for associative algebras and for Lie algebras (see below). For the relationship between Rota-Baxter operators of nonzero weight on Lie algebras and classical Yang-Baxter equation, see~\cite{BGN1,BGN2} and the references therein.}
  is defined to be a linear map $P : A\rightarrow A$ satisfying
\begin{equation}\label{eq:RB}
P(x)P(y)=P(P(x)y+xP(y)),\forall x,y\in A.\end{equation}
Rota-Baxter operators (on associative algebras) were introduced by G. Baxter to solve an analytic formula in probability~\cite{Baxter}. In fact, the  relation ~(\ref{eq:RB})
generalizes the integration by parts formula. It was G.-C. Rota who
realized its importance in combinatorics~\cite{Rota}. It has been related to many other areas in
mathematics (see \cite{Guo1,Guo2} and the references therein).
It has also appeared in several fields in mathematical
physics. For example, Rota-Baxter operators were found to
play a crucial role in the Hopf algebraic approach of Connes and
Kreimer to the renormalization of perturbative quantum field
theory (\cite{CK1,CK2}).

Completely independent of the above developments, the  relation ~(\ref{eq:RB}) in the context of Lie algebras has
it own motivation and developing history. In fact, a {\bf Rota-Baxter operator on a Lie algebra $(\frakg, [\,,\,])$}, namely a linear operator $P:\frakg\to \frakg$ such that
\begin{equation}\label{eq:rotaeq1}
[P(x),P(y)] = P([P(x),y]+[x,P(y)]), \text{for all } x, y\in \frakg,
\end{equation}
is also called the operator form of the classical Yang-Baxter equation due to
Semenov-Tian-Shansky's work on the fundamental study of the later (\cite{STS}).
Explicitly, let ${\frak g}$ be a Lie algebra and $r=\sum\limits_i a_i\otimes
b_i\in {\frak g}\otimes {\frak g}$. Recall that $r$ is called a solution of the {\bf
classical Yang-Baxter equation (CYBE)} in ${\frak g}$ if
\begin{equation}\label{eq:CYBE}[r_{12},r_{13}]+[r_{12},r_{23}]+[r_{13},r_{23}]=0\;\;{\rm
in}\;U({\frak g}), \end{equation} where $U({\frak g})$ is the universal
enveloping algebra of ${\frak g}$ and
\begin{equation}r_{12}=\sum\limits_i a_i\otimes b_i\otimes 1, \quad r_{13}=\sum\limits_i
a_i\otimes1\otimes b_i, \quad r_{23}=\sum\limits_i 1\otimes a_i\otimes
b_i.\end{equation}

The CYBE first arose in the study of
inverse scattering theory (\cite{FT1,FT2}). It can be regarded as a ``classical limit" of the quantum
Yang-Baxter equation (\cite{Belavin}). They play a crucial role in many fields such as symplectic geometry, integrable systems, quantum groups,
quantum field theory~(see~\cite{CP} and the references therein).
When $\frak g$ is finite dimensional, $r\in {\frak g}\otimes {\frak g}$ corresponds to a linear map
(classical $r$-matrix) due to the expression of $r$ under a basis of $\frak g$. It is Semenov-Tian-Shansky who
proved that the relation ~(\ref{eq:rotaeq1}) is equivalent to the tensor form (\ref{eq:CYBE}) of the CYBE
when the following two conditions are satisfied: (a) there exists a
nondegenerate symmetric invariant bilinear form on ${\frak g}$ and (b) $r$ is skew-symmetric.

Semenov-Tian-Shansky systematically studied the relations (\ref{eq:rotaeq1}) and (\ref{eq:CYBE}) in \cite{STS}. He gave classification results of the operators satisfying (\ref{eq:rotaeq1}) on semisimple Lie algebras in terms of certain linear maps associated to some specified subalgebras. But the explicit construction has not been obtained yet.
Such explicit classification of Rota-Baxter operators and solutions of CYBE under a basis is necessary since many applications in the related fields depend strongly on the explicit expression, whereas the other types of classification (such as in terms of subalgebras) might not be applied as conveniently.

More generally, both the skew-symmetric and non-skew-symmetric solutions of the classical Yang-Baxter equation in the semisimple Lie algebras have been considered (\cite{BD,St1,St2}).
In the non-skew-symmetric case, the relation ~(\ref{eq:rotaeq1}) for a Lie algebra is not
equivalent to the tensor form (\ref{eq:CYBE}) of CYBE over the same Lie algebra, but nevertheless gives solutions of the tensor form~(\mref{eq:CYBE}) of the CYBE over other related Lie algebras, in at least two ways.

First, it is shown in~\mcite{Bai0} that a Rota-Baxter operator on a Lie algebra $\mathfrak g$ satisfying the relation~(\ref{eq:rotaeq1}) gives a solution of CYBE in the semidirect sum Lie algebra $\mathfrak{g} \ltimes_{{\rm ad}^{\ast}} \mathfrak{g}^{\ast}$ from the dual representation of the adjoint representation (co-adjoint representation) of $\mathfrak{g}$.
 In particular the classification result in our study gives a family of solutions in
the 6-dimensional Lie algebra $\speciall \ltimes_{{\rm ad}^{\ast}} \speciall^{\ast}$.

Second, by~\cite{Agu}, a Rota-Baxter operator on $\mathfrak g$ gives pre-Lie algebra structure $A=A_{\mathfrak{g}}$ on the same underlying space of $\mathfrak g$. Pre-Lie
algebras are a class of nonassociative algebras coming
from the study of convex homogeneous cones, affine manifolds and
deformations of associative algebras and appeared in many
fields in mathematics and mathematical physics (see the survey article \cite{Bu} and the references therein).
It can be regarded as the algebraic structure
behind both the Rota-Baxter operator and the classical Yang-Baxter equation in Lie algebras~\cite{Bai0}.

Furthermore, a pre-Lie algebra $A$, by anti-symmetrizing, gives a Lie algebra $\mathfrak{g}(A)$ on its underlying space that has a representation on itself by the left multiplication $L$ of the pre-Lie algebra. This representation gives a solution of CYBE on in $\frak g(A)\ltimes_{L^*} \frak{g}(A)^*$.
In summary, a Rota-Baxter operator on a Lie algebra gives rise to a second Lie algebra structure $\mathfrak{g}(A)$ on the same underlying space and a solution of CYBE in the semidirect sum Lie algebra $\frak g(A)\ltimes_{L^*} \frak{g}(A)^*$. Note that $\frak g(A)\ltimes_{L^*} \frak{g}(A)^*$ is different from $\mathfrak{g} \ltimes_{{\rm ad}^{\ast}} \mathfrak{g}^{\ast}$!

Therefore, in order to study CYBE~(\ref{eq:CYBE}), among other purposes, it is important to explicitly determine the Rota-Baxter operators over a semisimple Lie algebra under certain canonical basis like Cartan-Weyl basis.
Unfortunately, it is not easy to carry it out for arbitrary semisimple Lie algebras. Thus, in this paper, we will focus on Rota-Baxter operators on $\speciall$. It is the simplest semisimple Lie algebra, yet has typical properties that might be used as a guide for more general investigations.
\smallskip

This paper is organized as follows. In Section 2, we give the explicit classification of Rota-Baxter operator on $\speciall$ under the Cartan-Weyl basis, by first reducing the classification problem to solution of a system of quadratic equations and then solving this system. In Section 3, we first specialize to the skew-symmetric Rota-Baxter operator $\speciall$ and give the explicit correspondence with the skew-symmetric solutions of CYBE in $\speciall$, as expected by Semenov-Tian-Shansky~\cite{STS}. Then for all Rota-Baxter operators on $\speciall$, we derive the induced solutions of CYBE in the 6-dimensional Lie algebra $\speciall \ltimes_{ad^{\ast}} \speciall^{\ast}$. In Section 4, we give the induced 3-dimensional pre-Lie algebras $A$ from these Rota-Baxter operators and the resulting solutions of CYBE in the 6-dimensional Lie algebras $\frak g(A)\ltimes_{L^*} \frak{g}(A)^*$. In Section 5, we give some conclusions and discussions based on the results in the previous sections.

\section{The Rota-Baxter operators on $\speciall$}
We first give some background notations and the statement of the classification theorem of Rota-Baxter operators on $\speciall$ in Section~\ref{sec:2.1}. The proof of the theorem is carried out in two parts. In Section~\ref{sec:2.2} the proof is reduced solving a system of quadratic equations. In Section~\ref{sec:2.3} the system of the quadratic equations is solved.

\subsection{Notations and the classification theorem}\label{sec:2.1}

Let $\speciall$ be the 3-dimensional special linear Lie algebra over the field of complex numbers $\mathbb{C}$. Let
$$
e = \left(\begin{array}{cc} 0&1\\0&0 \end{array}\right), \quad f = \left(\begin{array}{cc} 0&0\\1&0 \end{array}\right),\quad h = \left(\begin{array}{cc} 1&0\\0&-1 \end{array}\right)
$$
be the standard (Cartan-Weyl) basis of $\speciall$. Then we have
\begin{equation}\label{eq:product}
[h,e] = 2e, \quad [h,f] = -2f, \quad [e,f] = h.
\end{equation}
Thus a linear operator $P:\speciall\to \speciall$ is determined by
\begin{equation}\mlabel{eq:matrix}
\left( \begin{array}{c}
 P(e)\\
P(f) \\
P(h)\end{array} \right)  = \left( \begin{array} {ccc}
r_{11}&r_{12}&r_{13}\\
r_{21}&r_{22}&r_{23}\\
r_{31}&r_{32}&r_{33} \end{array}\right)   \left( \begin{array}{c}
e\\
f \\
h\end{array} \right),
\end{equation}
where $r_{ij}\in \CC, 1\leq i, j\leq 3$.
$P$ is a Rota-Baxter operator on $\speciall$ if the above matrix
$(r_{ij})_{3 \times 3}$ satisfies Eq.~(\ref{eq:rotaeq1}) for $x, y \in\{e, f, h\}$.

It follows from a direct check that $P$ is a  Rota-Baxter operator
if and only if $\lambda P$ is a Rota-Baxter operator
for $0\neq\lambda\in \CC$.
Thus the set $RB(\speciall)$ of Rota-Baxter operators on $\speciall$ carries
an action of $\CC^*:=\mathbb C\backslash \{0\}$ by scalar multiplication.
To determine all the Rota-Baxter operators on $\speciall$,
we only need to give a complete set of representatives of $RB(\speciall)$ under this action.

\begin{theorem} A complete set of representatives of $RB(\speciall)$ under the action of $\CC^*$ by scalar product consists of the 22 Rota-Baxter operators whose linear transformation matrices with respect to the Cartan-Weyl basis are listed below, where $a, b$ are non-zero complex numbers:

\begin{eqnarray*}
P_{1} = \left( \begin{array} {ccc}
0&0&0\\
0&0&1\\
0&0&0\end{array}\right), \quad  P_{2} = \left( \begin{array} {ccc}
0&0&0\\
0&0&0\\
0&0&0\end{array}\right), \quad P_{3} = \left( \begin{array} {ccc}
0&1&0\\
0&0&0\\
0&0&0\end{array}\right), \quad P_{4} = \left( \begin{array} {ccc}
0&0&0\\
0&0&0\\
0&0&1\end{array}\right), \\
P_{5} = \left( \begin{array} {ccc}
0&0&0\\
1&0&0\\
0&0&0\end{array}\right), \quad P_{6} = \left( \begin{array} {ccc}
0&0&0\\
1&0&a\\
0&0&0\end{array}\right), \quad P_{7} = \left( \begin{array} {ccc}
1&a&0\\
\frac{1}{a}&1&0\\
0&0&0\end{array}\right), \quad
P_{8} = \left( \begin{array} {ccc}
1&\frac{a^2}{16}&0\\
\frac{16}{a^2}&-3&-\frac{8}{a}\\
0&a&2\end{array}\right),\\
P_{9} = \left( \begin{array} {ccc}
0&0&1\\
0&0&0\\
0&0&0\end{array}\right), \quad P_{10} = \left( \begin{array} {ccc}
0&0&1\\
0&0&0\\
0&-2&0\end{array}\right), \quad P_{11} = \left( \begin{array} {ccc}
0&1&a\\
0&0&0\\
0&0&0\end{array}\right), \quad P_{12} = \left( \begin{array} {ccc}
0&1&a\\
0&0&0\\
0&-2a&0\end{array}\right),\\
P_{13} = \left( \begin{array} {ccc}
0&1&a\\
0&0&0\\
0&2a&2a^{2}\end{array}\right), \qquad P_{14} = \left( \begin{array} {ccc}
0&1&a\\
0&-4a^{2}& -4a^{3}\\
0&4a&4a^{2}\end{array}\right), \qquad P_{15} = \left( \begin{array} {ccc}
0&0&0\\
0&0&1\\
-2&0&0\end{array}\right),\\
P_{16} = \left( \begin{array} {ccc}
0&0&0\\
1&0&a\\
-2a&0&0\end{array}\right), \qquad
P_{17} = \left( \begin{array} {ccc}
0&0&0\\
1&0&a\\
2a&0& 2a^{2}\end{array}\right), \qquad
P_{18} = \left( \begin{array} {ccc}
-4a^{2}&0& -4a^{3}\\
1&0&a\\
4a&0& 4a^{2}\end{array}\right),
\end{eqnarray*}

\begin{eqnarray*}
P_{19} = \left( \begin{array} {ccc}
1&  -\frac{3a^{2}}{4}&a\\
- \frac{4}{27a^{2}}& -\frac{1}{3}& 0\\
-\frac{8}{9a}&0& -\frac{2}{3}\end{array}\right),\quad
P_{20} = \left( \begin{array} {ccc}
a&0&- \frac{a^2}{2}\\
0& -a& - \frac{1}{2}\\
1&a^2&0\end{array}\right), \quad P_{21} = \left( \begin{array} {ccc}
a&4a^{3}&0\\
-\frac{1}{4a}&-a&0\\
1&4a^{2}&0\end{array}\right),
\end{eqnarray*}

\begin{eqnarray*}
 P_{22} = \left( \begin{array} {ccc}
-\frac{1}{4b}&a&-\frac{1+16ab^{3}}{16b^2}\\
b&-4ab^{2}&\frac{1+16ab^{3}}{4}\\
1&-4ab&\frac{1+16ab^{3}}{4b}\end{array}\right).
\end{eqnarray*}
\mlabel{thm:rbo}
\end{theorem}

In Sections~2.2 and 2.3, we prove Theorem~\mref{thm:rbo} by first reducing the matrix equation from the Rota-Baxter relation to
a system of 9 quadratic equations and then solving these quadratic equations.

\subsection{Reduction to quadratic equations}\label{sec:2.2}
By skew-symmetry, in order to show that $P$ is a Rota-Baxter operator, we only need to check
\begin{eqnarray*}
[P(e),P(f)] &=& P([P(e),f] + [e,P(f)]),\\{}
[P(e),P(h)] &=& P([P(e),h] + [e,P(h)]),\\{}
[P(f),P(h)] &=& P([P(f),h] + [f,P(h)]).
\end{eqnarray*}
It follows from Eq.~(\ref{eq:product}) that
\begin{eqnarray}
\lefteqn{[P(e),P(f)] = [r_{11}e+r_{12}f+r_{13}h,  r_{21}e+r_{22}f+ r_{23}h]} \notag\\
&=& 2(r_{13}r_{21} - r_{11}r_{23}) e + 2(r_{12}r_{23} - r_{13}r_{22}) f   + (r_{11}r_{22} - r_{12}r_{21}) h,\label{eq:lefthand}
\end{eqnarray}
while
\begin{eqnarray}
\lefteqn{P([P(e), f] + [e,P(f)] ) = P ((r_{11} + r_{22})h  -2 r_{23}e -2r_{13} f)} \notag\\
&=& (r_{11} + r_{22})r_{31} e  +  (r_{11} + r_{22}) r_{32} f   + (r_{11} + r_{22})r_{33} h -2 r_{23} r_{11} e -2 r_{23} r_{12} f  \notag\\
&&-2 r_{23} r_{13} h   -2 r_{13} r_{21} e  - 2 r_{13} r_{22} f  - 2 r_{13} r_{23} h \notag\\
 &=& \big((r_{11} + r_{22})r_{31} - 2r_{23}r_{11} - 2r_{13}r_{21}\big)e  + \big( (r_{11} + r_{22}) r_{32} -2 r_{23} r_{12} - 2 r_{13} r_{22} \big) f \notag\\
 && + \big((r_{11} + r_{22})r_{33} -2 r_{23} r_{13}  - 2 r_{13} r_{23} \big) h.\label{eq:righthand}
\end{eqnarray}
Comparing the coefficients in Eq.~(\ref{eq:lefthand}) and Eq.~(\ref{eq:righthand}), we have
\begin{eqnarray}
&4 r_{13} r_{21} =( r_{11} + r_{22} ) r_{31},&\label{eq:c1}\\
&4 r_{12} r_{23} = (r_{11} + r_{22}) r_{32},&\label{eq:c2}\\
&4 r_{23} r_{13} = (r_{11} + r_{22})r_{33} + r_{12}r_{21} - r_{11}r_{22}.&\label{eq:c3}
\end{eqnarray}
Similarly, from
\begin{equation*}
[P(e),P(h)] = P([P(e),h] + [e,P(h)]), \quad
[P(f),P(h)] = P([P(f),h] + [f,P(h)])
\end{equation*}
we obtain the following six equations
\begin{eqnarray}
&2r_{13}r_{31} = 2r_{12}r_{21} + r_{32}r_{31} - 2r_{11}^{2}, &\label{eq:c4}\\
&4r_{12}r_{33} = 2r_{12}r_{22} + r_{32}^{2} - 2 r_{11} r_{12} + 2 r_{13}r_{32},&\label{eq:c5}\\
&r_{11} r_{32} = 2r_{12} r_{23} + r_{32}r_{33} - 2(r_{11} + r_{33}) r_{13} + r_{12}r_{31},&\label{eq:c6}\\
&4r_{21}r_{33} = 2r_{23}r_{31} + 2r_{21}r_{11} + r^{2}_{31} - 2r_{22}r_{21},&\label{eq:c7}\\
&2r_{22}^{2} =- 2 r_{23} r_{32} +   2r_{21}r_{12} + r_{31}r_{32}, &\label{eq:c8}\\
&r_{21} r_{32} - r_{22} r_{31} = 2(r_{22} + r_{33}) r_{23}  - 2r_{21} r_{13} - r_{31}r_{33}.&\label{eq:c9}
\end{eqnarray}

\subsection{Solving the quadratic equations}\label{sec:2.3}
To solve the quadratic equations (\mref{eq:c1})-(\mref{eq:c9}), we distinguish the two cases depending on whether or not $r_{31}=0$.
\smallskip

\noindent
{\bf Case (\textrm{I}) $r_{31}=0$: } Then Eq.~(\ref{eq:c1}) implies $r_{13}r_{21} =0$. There are three subcases: (a) $r_{13} = r_{21}=0$; (b) $r_{13}\neq0, r_{21}=0$, and  (c) $r_{13}=0,r_{21} \neq 0$.
    \begin{enumerate}
    \item[(a)] Assume $r_{13}=0$ and $r_{21} =0$. Then Eq.~(\ref{eq:c4}) implies $r_{11}=0$. Also Eq.~(\ref{eq:c9}) implies $r_{23}(r_{22}+r_{33})=0$.
        \begin{enumerate}
          \item[(a1)] If $r_{23} \neq 0$, then $r_{22} = -r_{33}$. Eq. (\ref{eq:c3}) implies $r_{33}^{2}=0$ and then $r_{22}=r_{33}=0$. Each of Eq. (\ref{eq:c2}) and Eq. (\ref{eq:c6}) gives $r_{12} =0$ and each of Eq. (\ref{eq:c5}) and Eq. (\ref{eq:c8}) gives $r_{32}=0$. Taking $r_{23} = 1$, we obtain $P_{1}$.
          \item[(a2)] If $r_{23} =0$, then Eq. (\ref{eq:c8}) implies $r_{22}=0$. Then Eqs. (\ref{eq:c5}) and (\ref{eq:c6}) imply
            $$
            r_{32}^2 = 4 r_{12}r_{33}, \quad r_{32} r_{33}=0.
            $$
              \begin{enumerate}
              \item[(a21)] If $r_{33}=0$, then $r_{32}=0$. Further, if $r_{12} =0$, we get $P_{2}$; while if $r_{12} \neq 0 $ and taking $r_{12} =1$, we get $P_{3}$.
              \item[(a22)] If $r_{33}\neq 0$, then $r_{32} = r_{12} =0$. Taking $r_{33} =1$, then we get $P_{4}$.
              \end{enumerate}
         \end{enumerate}
    \item[(b)] Assume $r_{13}=0$ and $r_{21} \neq 0$.  Then Eq. (\ref{eq:c7}) implies $r_{21}(r_{11}-r_{22} -2r_{33})=0$  and then $r_{33} = \displaystyle \frac{r_{11}-r_{22}}{2}$. Eq. (\ref{eq:c4}) implies $r_{11}^{2} = r_{12} r_{21}$ and then
$$
r_{11}=0 \Longleftrightarrow r_{12} =0.
$$
So we have
          \begin{enumerate}
          \item[(b1)] If $r_{11} = r_{12} =0$, then Eq. (\ref{eq:c5}) gives $r_{32}=0$,  Eq. (\ref{eq:c7}) gives $r_{22}=0$, and Eq. (\ref{eq:c7}) gives $4r_{21}r_{33}=0$ and then $r_{33}=0$. Finally, $r_{23}$ can be arbitrary. Taking $r_{21}=1$ and also $r_{23}=a$ when $r_{23} \neq 0$, we get $P_{5}$ and $P_{6}$.
          \item[(b2)]  If $r_{11} \neq 0$ and $r_{12} \neq 0$, taking $r_{11} =1$, then Eq.~(\ref{eq:c4}) implies $2 = 2 r_{12} r_{21}$ and then $r_{21}= \displaystyle \frac{1}{r_{12}}$. Since $r_{33} = \displaystyle \frac{1-r_{22}}{2}$, Eq.~(\ref{eq:c3}) implies
              $$
               r_{22}^{2} + 2r_{22} - 3=0.
              $$
              Hence, $r_{22} =1$ or $r_{22} =-3$.
                  \begin{enumerate}
                    \item[(b21)] If $r_{22} =1$, then $r_{33}=0$. Eq. (\ref{eq:c5}) implies $r_{32}^{2}=0$ and then $r_{32}=0$ and Eq. (\ref{eq:c2}) implies $r_{23}=0$. Denoting $r_{12}=a$, we get $P_{7}$.
                    \item[(b22)] If $r_{22}=-3$, then $r_{33}=2$. Thus Eq. (\ref{eq:c5}) implies $16 r_{12} = r_{32}^2 $ and then $r_{12} =\displaystyle \frac{r_{32}^{2}}{16}$, $r_{32} \neq 0$.  Eq. (\ref{eq:c8}) implies $r_{23} r_{32} = -8$ and then $r_{23} = \displaystyle -\frac{8}{r_{32}}$. Taking $r_{32}=a$, we get $P_{8}$.
                  \end{enumerate}
          \end{enumerate}
    \item[(c)] Assume $r_{13} \neq 0$, $r_{21} =0$. Then Eq. (\ref{eq:c4}) implies $r_{11}=0$.
          \begin{enumerate}
          \item[(c1)] If $r_{12}=0$, taking $r_{13}=1$, then Eq.~(\ref{eq:c5}) implies $r_{32}(2+r_{32}) = 0$. Thus $r_{32}=0$ or $r_{32}=-2$.
               \begin{enumerate}
               \item[(c11)] If $r_{32}=0$, then Eq.~(\ref{eq:c6}) implies $r_{33}=0$ and Eq.~(\ref{eq:c8}) implies $r_{22}=0$. Thus Eq. (\ref{eq:c3}) gives $r_{23}=0$. We obtain $P_{9}$.
               \item[(c12)] If $r_{32}=-2$, then Eq.~(\ref{eq:c2}) implies $r_{22}=0$. Further Eq.~(\ref{eq:c3}) implies $r_{23}=0$ and Eq. (\ref{eq:c6}) implies $r_{33}=0$. Thus we obtain $P_{10}$.
               \end{enumerate}
          \item[(c2)] If $r_{12} \neq 0$, then take $r_{12} =1$ and $r_{13}=a$.
\begin{enumerate}
\item[(c21)] If $r_{23}=0$, then Eq.~(\ref{eq:c8}) implies $r_{22}=0$. Eqs. (\ref{eq:c5}) and (\ref{eq:c6}) imply
$$
-2ar_{32} - r_{32}^{2} + 4  r_{33} = 0 ,\quad  (2 a - r_{32}) r_{33}=0.
$$
\begin{enumerate}
\item[(c211)] If $r_{33}=0$, then $r_{32}=0$ or $r_{32}=-2a$. We obtain $P_{11}$ or $P_{12}$ respectively.
\item[(c212)] If $r_{33} \neq 0$, then $r_{32}=2a$ and $r_{33} =2a^{2}$. We get $P_{13}$.
\end{enumerate}
\item[(c22)]  If $r_{23} \neq 0$, then Eq. (\ref{eq:c9}) implies $r_{22}+r_{33}=0$. Further, Eq. (\ref{eq:c3}) and Eq. (\ref{eq:c8}) imply
$$
4 a r_{23} + r_{33}^{2}=0, \quad -2 r_{23} r_{32} - 2 r_{33}^{2}=0.
$$
Hence $r_{32} =4a$. Eq. (\ref{eq:c5}) implies $r_{33} = 4a^{2}$ and then $r_{22} = -4a^{2}$. Also Eq.~(\ref{eq:c2}) implies $r_{23} = -4a^{3}$. Thus we get $P_{14}$.
\end{enumerate}
          \end{enumerate}
    \end{enumerate}

\noindent
{\bf Case (\textrm{II}) $r_{31} \neq 0$}. Then we distinguish two subcases: (d) $r_{32}=0$ and (e) $r_{32}\neq 0$.
\begin{enumerate}
\item[(d)] Assume $r_{32}=0$. Then Eq. (\ref{eq:c2}) implies $r_{12}r_{23}=0$. There are three different subcases to be considered: $r_{12} = r_{23} =0; r_{12}=0, r_{23} \neq 0; r_{12} \neq 0, r_{23} =0$.
    \begin{enumerate}
    \item[(d1)] If $r_{12}=0$ and $r_{23}=0$, then Eq. (\ref{eq:c8}) implies $r_{22}=0$ and Eq. (\ref{eq:c3}) implies $r_{11}r_{33}=0$.
              \begin{enumerate}
              \item[(d11)] If $r_{11}=0$, then Eq. (\ref{eq:c4}) implies $r_{13}=0$. Then Eqs. (\ref{eq:c7}) and (\ref{eq:c9}) imply $r_{31}=0$ and $r_{33}=0$. It is in contradiction with $r_{31} \neq 0$.
              \item[(d12)] If $r_{11}\neq 0$, then $r_{33}=0$.  then Eq. (\ref{eq:c6}) implies $r_{13}=0$.  Then Eq. (\ref{eq:c4}) implies $r_{11}=0$, a contradiction with $r_{11} \neq 0$.
              \end{enumerate}
     So $r_{12}$ and $r_{23}$ can not be zero at the same time.
     \item[(d2)] If $r_{12}=0$ and $r_{23} \neq 0$, then Eq.~(\ref{eq:c8}) implies $r_{22}=0$.
              \begin{enumerate}
              \item[(d21)] If $r_{21}=0$, then Eq.~(\ref{eq:c1}) implies $r_{11}=0$. Eq.~(\ref{eq:c4}) implies $r_{13}=0$. Eqs.~(\ref{eq:c7}) and (\ref{eq:c9}) imply $r_{31} = -2r_{23}$ and $r_{33}=0$. Taking $r_{23}=1$, we get $P_{15}$.
              \item[(d22)] If $r_{21} \neq 0$ and $r_{13}=0$, then Eq.~(\ref{eq:c1}) implies $r_{11}r_{31} =0$ and then $r_{11}=0$. Then Eqs. (\ref{eq:c7}) and (\ref{eq:c9}) imply
                 $$
                   2r_{23}r_{31} + r_{31}^{2} - 4 r_{21} r_{33} = 0, \quad  -2r_{23}r_{33} + r_{31} r_{33}=0.
                 $$
                Then we have $r_{33}=0$ and $r_{31}=-2r_{23}$, or $r_{31}=2r_{23}$ and $r_{33}= \displaystyle \frac{2r_{23}^{2}}{r_{21}}$. Taking $r_{21} =1$ and $r_{23} =a$, we obtain $P_{16}$ and $P_{17}$ respectively.
              \item[(d23)] If $r_{21} \neq 0$ and $r_{13} \neq 0$, then Eq.~(\ref{eq:c6}) implies $r_{11} + r_{33}=0$.
                       \begin{enumerate}
                       \item[(d231)] If $r_{11} = r_{33} =0$, then Eq. (\ref{eq:c1}) implies $r_{13}=0$ and gives a contradiction with $r_{13} \neq 0$.
                       \item[(d232)] If $r_{33} = - r_{11} \neq 0$, then Eqs. (\ref{eq:c4}) and (\ref{eq:c5}) imply
                          $$
                         r_{11}^{2} + 4 r_{13} r_{23}=0, \quad  2 r_{11}^2 + 2 r_{13} r_{31} =0.
                          $$
                         Hence, we have $r_{31}=4r_{23}$. Then Eqs. (\ref{eq:c1}), (\ref{eq:c3}) and (\ref{eq:c7}) imply
                        $$
                         r_{13}r_{21} = r_{11}r_{23}, \quad r_{11}^{2} = -4r_{13}r_{23}, \quad r_{11}r_{21} = -4r_{23}^{2}.
                        $$
                       Then $r_{11} = \displaystyle \frac{-4r_{23}^{2}}{r_{21}}$,  $r_{13} = \displaystyle \frac{-4r_{23}^{3}}{r_{21}^{2}}$. Taking $r_{21}=1$ and $r_{23}=a$, we obtain $P_{18}$.
                       \end{enumerate}
              \end{enumerate}
    \item[(d3)]   If $r_{12} \neq 0$ and $r_{23} =0$, then Eq. (\ref{eq:c5}) implies $r_{33} = \displaystyle \frac{r_{22} - r_{11}}{2}$.
              \begin{enumerate}
              \item[(d31)] If $r_{13}=0$, then Eq. (\ref{eq:c6}) implies $r_{31}=0$ and gives a contradiction. Hence, $r_{13} \neq 0$.
              \item[(d32)] If $r_{22} =0$, then Eq. (\ref{eq:c8}) implies $r_{21}=0$. Then Eq. (\ref{eq:c1}) implies $r_{11}=0$ and Eq. (\ref{eq:c3}) implies $r_{13}=0$, giving a contradiction. So $r_{22} \neq 0$.
              \item[(d33)] If $r_{11} = 0$, then Eq. (\ref{eq:c5}) implies $r_{33} = \displaystyle \frac{r_{22}}{2}$. Thus  Eqs.~(\ref{eq:c3}) and (\ref{eq:c8}) imply
                $$
                 r_{12}r_{21} + \frac{r_{22}^{2}}{2} =0, \quad  r_{12}r_{21} -r_{22}^{2}=0.
                 $$
                  Hence $r_{22}=0$, giving a contradiction. Therefore, $r_{11} \neq 0$.
              \item[(d34)] If $r_{21} = 0$, then Eq.~(\ref{eq:c7}) implies that $r_{31} =0$ and gives a contradiction. Hence, $r_{21} \neq 0$.
              \end{enumerate}
              In summary we have $r_{11}, r_{12}, r_{13}, r_{21}, r_{22},r_{31} \neq 0$, $r_{23} = r_{32}=0$ and $r_{33} = \displaystyle \frac{r_{22}-r_{11}}{2}$. Taking $r_{11} =1$ and $r_{13} =a$, then Eqs. (\ref{eq:c2}) and (\ref{eq:c9}) become
             $$
              1+2r_{22} = r_{22}^{2} +2r_{12}r_{21},\quad  r_{12} r_{21} = r_{22}^{2}.$$
             Hence, we have
          \begin{eqnarray*}
           3r_{22}^{2} -2 r_{22} - 1 =0,
             \end{eqnarray*}
             and then $r_{22} = 1$ or $r_{22} = -\displaystyle \frac{1}{3}$.
            If $r_{22} =1$, then Eq.~(\ref{eq:c7}) implies $r_{31}=0$ and gives a contradiction. If $r_{22} = - \displaystyle \frac{1}{3}$,
           then $r_{12} = \displaystyle -\frac{3a^{2}}{4}, r_{21} = -\frac{4}{27a^{2}}, r_{31} = -\frac{8}{9a}$ and $r_{33}=\displaystyle - \frac{2}{3}$. Thus we get $P_{19}$.
    \end{enumerate}
\item[(e)] Assume $r_{32} \neq 0$. Take $r_{31}=1$.
     \begin{enumerate}
          \item[(e1)] If $r_{21} =0$, then Eq. (\ref{eq:c1}) implies $r_{22} = -r_{11}$, Eq. (\ref{eq:c7}) implies $r_{23} = \displaystyle-\frac{1}{2}$ and Eq. (\ref{eq:c2}) implies $r_{12} =0$. Further, Eq. (\ref{eq:c5}) implies $r_{32} (2r_{13} + r_{32})=0$ and hence $r_{13} = \displaystyle-\frac{r_{32}}{2}$. Moreover, Eq. (\ref{eq:c9}) implies $r_{33}=0$ and Eq. (\ref{eq:c3}) implies $r_{32} = r_{11}^{2}$. Taking $r_{11} =a$, we conclude
          \begin{eqnarray*}
          &r_{11} = a, \quad r_{12}=0, \quad r_{13}=\displaystyle-\frac{a^2}{2}, \quad r_{21}=0, \quad r_{22}=-a, \quad r_{23} =\displaystyle -\frac{1}{2}, \quad r_{33}=0.& \\
          \end{eqnarray*}
          Thus we obtain $P_{20}$.
          \item[(e2)] If $r_{21} \neq 0$ and $r_{13} =0$, then Eq. (\ref{eq:c1}) implies $r_{22} = -r_{11}$. Eqs. (\ref{eq:c2}), (\ref{eq:c3}) and (\ref{eq:c4}) imply
            \begin{eqnarray*}
            4 r_{12} r_{23}=0, \quad  -r_{11}^{2} - r_{21} r_{12} = 0, \quad -r_{32} + 2 r_{11}^2 - 2 r_{21} r_{12}=0.
            \end{eqnarray*}
            Thus we have $r_{12} = \displaystyle-\frac{r_{32}}{4r_{21}} \neq 0$. Hence $r_{23} =0$ and $ r_{32} = 4 r_{11}^{2}$. Taking $r_{11} =a$, then we have
            \begin{eqnarray*}
            r_{22}=-a, \quad r_{23}=0,\quad r_{32} = 4a^{2}.
            \end{eqnarray*}
            Eq. (\ref{eq:c5}) implies the same equation as Eq. (\ref{eq:c9}) and Eq. (\ref{eq:c6}) implies the same equation as Eq. (\ref{eq:c7}), so we have the following two equations
                   $$
                   a + 4 a^2 r_{21} + r_{33}=0, \quad 4a + \frac{1}{r_{21}} - 4 r_{33}=0.
                   $$
                 Thus we have $r_{21} = -\displaystyle\frac{1}{4a}$ and $r_{33}=0$. Then $r_{12} = 4a^{3}$. That is
              \begin{eqnarray*}
              &\quad r_{12} = 4a^{3}, \quad r_{21}=\displaystyle -\frac{1}{4a}, \quad r_{22}=-a,\quad r_{32} = 4a^{2}, \quad r_{33}=0.&
              \end{eqnarray*}
            In summary we obtain $P_{21}$.
            \item[(e3)] If $r_{21} \neq 0, r_{13} \neq 0$. Then Eq. (\ref{eq:c1}) and $r_{13}r_{21} \neq 0$ imply $r_{11} + r_{22} \neq 0$. Therefore, Eq. (\ref{eq:c2}) implies $r_{32} (r_{11} + r_{22}) = 4 r_{12} r_{23}$ and then $r_{12} \neq 0, r_{23} \neq 0$. Eq. (\ref{eq:c1}) and (\ref{eq:c2}) imply $r_{11} =4 r_{13} r_{21} - r_{22}$, $r_{32} = \displaystyle \frac{r_{12}r_{23}}{r_{13} r_{21}}$. Then  Eq. (\ref{eq:c3}) implies
                $$
                r_{33} = -\frac{r_{12}}{4r_{13}} + r_{22} + \frac{r_{23}}{r_{21}} - \frac{r_{22}^{2}}{4r_{13}r_{21}}.
                $$
                Thus, Eq. (\ref{eq:c7}) implies
                $$
                r_{23} = \frac{r_{13} + r_{12} r_{21} + 8 r_{13}^{2} r_{21}^{2} - 8 r_{13} r_{21} r_{22} + r_{22}^{2}}{2r_{13}}.
                $$
             Hence,
             \begin{eqnarray*}
             &r_{32} =  \frac{r_{12}}{r_{13} r_{21}} \big( \frac{r_{13} + r_{12} r_{21} + 8 r_{13}^{2} r_{21}^{2} - 8 r_{13} r_{21} r_{22} + r_{22}^{2}}{2r_{13}}   \big),&\\
             &r_{33} = -\frac{r_{12}}{4r_{13}} + r_{22} + \frac{r_{13} + r_{12} r_{21} + 8 r_{13}^{2} r_{21}^{2} - 8 r_{13} r_{21} r_{22} + r_{22}^{2}}{2r_{13}r_{21}} - \frac{r_{22}^{2}}{4r_{13}r_{21}}.&
             \end{eqnarray*}
             Then Eq.~(\ref{eq:c4}) becomes
             $$ 2 r_{13} - 2 r_{12} r_{21} + 2 (-4 r_{13} r_{21} + r_{22})^{2} =\frac{r_{12} \big(r_{13} + r_{12} r_{21} + 8 r_{13}^{2} r_{21}^{2} - 8 r_{13} r_{21} r_{22} + r_{22}^2 \big)}{2 r_{13}^{2} r_{21}},$$
             which simplifies to
             $$r_{13} + r_{12} r_{21} + 16 r_{13}^2 r_{21}^2 - 8 r_{13} r_{21} r_{22} + r_{22}^2=0.$$
             Then $r_{22} = 4 r_{13} r_{21} \pm \sqrt{-r_{13} - r_{12} r_{21}}$ and
             \begin{eqnarray*}
             &r_{23} = -4r_{13} r_{21}^{2}, \quad r_{33} = \frac{1}{4 r_{21}} - 4 r_{13} r_{21} \mp \sqrt{-r_{13} - r_{12} r_{21}},&\\
             &r_{32} = -4r_{12}r_{21}, \quad r_{11} = \mp \sqrt{-r_{13} - r_{12} r_{21}}.&
             \end{eqnarray*}
             Thus each of the Eqs. (\ref{eq:c1})-(\ref{eq:c9}) is either trivial or implies
             \begin{eqnarray*}
             1 - 16 r_{13} r_{21}^2 - 16 r_{12} r_{21}^3 \mp 8 r_{21} \sqrt{-r_{13} - r_{12} r_{21}}=0.
             \end{eqnarray*}
             Therefore, we have,
             \begin{eqnarray*}
             (r_{13} + \frac{1}{16r_{21}^2} +r_{12}r_{21} )^2=0.
             \end{eqnarray*}
             Hence, $r_{13} =-\frac{1}{16r_{21}^2} -r_{12}r_{21}$. Thus
             \begin{eqnarray*}
             &r_{11} =\mp \frac{1}{4\sqrt{r_{21}^{2}}}, \quad r_{33} =\frac{1}{2r_{21}} \mp \frac{1}{4\sqrt{r_{21}^{2}}} + 4r_{12}r_{21}^{2},&\\
             &r_{23} = \frac{1}{4} + 4r_{12}r_{21}^{3}, \quad r_{22} = -\frac{1}{4r_{21}} \pm \frac{1}{4 \sqrt{r_{21}^{2}}}- 4r_{12}r_{21}^{2}.&
             \end{eqnarray*}
             Then Eq. (\ref{eq:c5}) gives
             $
             1 \mp \frac{1}{\sqrt{r_{21}^2}} r_{21} =0,
             $
             that is $\sqrt{r_{21}^{2}} =\pm r_{21}$. Therefore, we get
             \begin{eqnarray*}
             r_{11} = -\frac{1}{4r_{21}},\quad r_{13} = -\frac{1}{16r_{21}^{2}} - r_{12} r_{21}, \quad r_{22} = -4r_{12}r_{21}^{2},\\
             r_{23} = \frac{1}{4} + 4r_{12}r_{21}^{3}, \quad
             r_{32}= -4 r_{12}r_{21},  \quad r_{33} = \frac{1}{4r_{21}} + 4r_{12}r_{21}^{2}.
             \end{eqnarray*}
             Since for any $r_{21} \neq 0$, we have $r_{21}= \pm \sqrt{r_{21}^{2}}$, we can take $r_{21}=b$. Then taking $r_{12}=a$, we obtain $P_{22}$.
          \end{enumerate}
\end{enumerate}

One checks that all the above prescriptions for $P_i$ are indeed solutions of Eqs.~(\mref{eq:c1})-(\ref{eq:c9}). The proof of Theorem~\mref{thm:rbo} is completed.

\section{Rota-Baxter operators on $\speciall$ and solutions of CYBE}

We now give the first application of our classification theorem, namely using the Rota-Baxter operators obtained in Theorem~\mref{thm:rbo} to derive solutions of CYBE in the Lie algebras $\speciall$ and $\speciall \ltimes_{{\rm ad}^{\ast}} \speciall^{\ast}$.

\subsection{Skew-symmetric solutions of CYBE in $\speciall$}

Let $\mathfrak{g}$ be a Lie algebra with a nondegenerate invariant bilinear form.
Then we can identify $\mathfrak{g}$ with $\mathfrak{g}^{\ast}$ via the bilinear
form. Thus
$$\mathfrak{g} \otimes \mathfrak{g} \cong \mathfrak{g}
\otimes \mathfrak{g}^{\ast} \cong {\rm End}({\mathfrak{g}}).$$
We would like to emphasize that the explicit correspondence (in particular, between the basis of $\frak g$ and $\frak g^*$) holds
only when we take an orthonormal
basis of $\mathfrak{g}$. Under the orthonormal basis, $r \in \mathfrak{g}^{2}$ is skew-symmetric if and only if
the matrix form of its corresponding operator is skew-symmetric and furthermore,
$r$ is a skew-symmetric solution of CYBE
in $\mathfrak{g}$ if and only if its corresponding operator is a {\em skew-symmetric} Rota-Baxter operator on $\mathfrak{g}$ in the sense that the corresponding matrix is skew-symmetric.

It is easy to check that
$$\alpha = \frac{\sqrt{2}i}{4}(e-f), \quad \beta = \frac{\sqrt{2}}{4}(e+f), \quad \gamma = \frac{\sqrt{2}}{4} h,$$
is an orthonormal basis of $\speciall$ with respect to the Killing form.

\begin{theorem}
Under the orthonormal basis $\{\alpha, \beta, \gamma\}$, the matrices of the non-zero skew-symmetric Rota-Baxter operators on $\speciall$ are given by
$$
P_{10} \leftrightarrow \left( \begin{array} {ccc}
0&0&i\\
0&0&1\\
-i&-1&0\end{array}\right), \quad
P_{15} \leftrightarrow \left( \begin{array} {ccc}
0&0&-i\\
0&0&1\\
i&-1&0\end{array}\right), \quad
P_{20} \leftrightarrow \left( \begin{array} {ccc}
0&a i& \frac{1-a^2}{2} i\\
-a i& 0 & - \frac{1+a^{2}}{2}\\
\frac{a^{2}-1}{2}i&\frac{1+a^2}{2}&0\end{array}\right).
$$
The corresponding skew-symmetric solutions of  CYBE in $\speciall$ are given by
\begin{eqnarray*}
r_{1}&=& k(e \otimes h -h \otimes e),\\
r_{2}&=& k(f \otimes h - h \otimes f),\\
r_{3}&=& k \left (a(f \otimes e - e \otimes f) + \frac{1}{2}(h \otimes f - f \otimes h)
+ \frac{a^{2}}{2}(h \otimes e - e \otimes h)\right),
\end{eqnarray*}
where $k,a \in \mathbb{C}\backslash \{0\}$. Moreover, together with the zero solution, they correspond precisely to all the skew-symmetric solutions of CYBE in $\speciall$ given by Belevin and Drinfeld in~\mcite{BD}.
\mlabel{thm:cybe}
\end{theorem}

In fact, by Theorem~\mref{thm:rbo}, it is straightforward to check that under the orthonormal basis $\{\alpha, \beta, \gamma\}$,
only three types of non-zero Rota-Baxter operators whose corresponding matrices are skew-symmetric: $P_{10},P_{15},P_{20}$. Moreover,
the corresponding matrices under the orthonormal basis $\{\alpha, \beta, \gamma\}$ are given by the three matrices in the theorem and the corresponding skew-symmetric solutions of CYBE follow immediately.

On the other hand, we recall the following classification result on the skew-symmetric solution of CYBE in $\speciall$ given by Belevin and Drinfeld.
\begin{theorem}\mcite{BD}
Let $x,y \in \mathfrak{g}$ such that $[x, y] = y$. Then $r = x \otimes y - y \otimes x$ is a skew-symmetric solution of (\mref{eq:CYBE}). Furthermore, this construction gives all nonzero skew-symmetric solutions for $\speciall$.
\end{theorem}

Of course, the skew-symmetric solutions of CYBE in $\speciall$ given by the three matrices in Theorem~\mref{thm:cybe} and those given by Belevin and Drinfeld in terms of 2-dimensional non-abelian subalgebras of $\speciall$ should coincide. To be precise, we next give an explicit 1-1 correspondence between the skew-symmetric solutions of CYBE in $\speciall$ given by the three matrices in Theorem~\mref{thm:cybe} and those given by Belevin and Drinfeld in terms of 2-dimensional non-abelian subalgebras of $\speciall$.

Let ~$x = x_{1}e + x_{2}f +x_{3}h$ and $y = y_{1} e + y_{2} f + y_{3} h$ in $\speciall$ be such that $[x,y] =y$. Then we have
\begin{equation}\mlabel{eq:solution}
\left\{\begin{array}{l} (2 x_{3} -1) y_{1} - 2 x_{1} y_{3} =0 \\
2 x_{2} y_{3} - (2 x_{3} + 1) y_{2} =0 \\
x_{1} y_{2} - x_{2} y_{1} - y_{3} =0, \end{array}\right.
\end{equation}
regarded as a linear system with variables $y_1, y_2, y_3$.
Since the determinant of the coefficient matrix is
$$
\left|\begin{array}{ccc} (2x_{3}-1)&0&-2x_{1}\\0&-(2x_{3}+1)&2x_{2}\\-x_{2}&x_{1}&-1 \end{array}\right| =
(2x_{3}-1)(2x_{3}+1) + 4x_{1}x_{2},
$$
the system of Eqs. (\ref{eq:solution}) has non-zero solutions if and if only $(2x_{3}-1)(2x_{3}+1) + 4x_{1}x_{2} =0$.

\noindent
{\bf Case I: Assume that $x_{1} =0$.} Then $x_{3}$ should be $\displaystyle \pm \frac{1}{2}$. If $x_{3} =\displaystyle \frac{1}{2}$, we have $y_{2} = -x_{2}^{2} y_{1}, y_{3} = - x_{2} y_{1}$.
Thus $\displaystyle x = x_{2}f + \frac{h}{2}$ and $y = y_{1}e -x_{2}^{2} y_{1} f  -x_{2} y_{1} h.$ We obtain
$$
\displaystyle x \otimes y - y \otimes x = y_{1}
\big( x_{2} (f \otimes e - e \otimes f)
+ \frac{1}{2} (h \otimes e - e \otimes h) + \frac{x^{2}_{2}}{2}
 (h \otimes f - f \otimes h) \big).
$$
If $x_2\neq 0$, by taking $k=y_{1}$ and $a =x_{2}$ in the above equation, we obtain $r_3$. If $x_{2} = 0$, we obtain $r_{1}$
by taking $k=\displaystyle -\frac{1}{2}y_1$,

If $x_{3} = \displaystyle -\frac{1}{2}$,
we have $y_{1} = y_{3} =0$. The non-zero solutions of Eq.~(\mref{eq:solution}) give
$\displaystyle x = x_{2} f - \frac{1}{2} h$ and $ y =  y_{2}f.$ This gives
\begin{eqnarray*}
x \otimes y - y \otimes x = \frac{1}{2}y_{2} (f \otimes h - h \otimes f) =
r_{2}
\end{eqnarray*}
by taking $k=\displaystyle \frac{1}{2}y_2$.
\smallskip

\noindent
{\bf Case II: Assume that $x_{1} \neq 0$.} We have $x_{2} = \displaystyle \frac{1-4x_{3}^{2}}{4x_{1}}$. Then
$$
\displaystyle x= x_{1} e + \frac{1-4x_{3}^{2}}{4x_{1}} f  + x_{3} h,\quad y = y_{1} \Big( e -\frac{(2x_{3}-1)^{2}}{4x_{1}^{2}} f  +
  \frac{2x_{3}-1}{2x_{1}} h \Big).
$$
Thus we obtain
$$ x \otimes y - y \otimes x =
 y_1\left(\frac{2x_{3}-1}{2x_{1}}(f \otimes e - e \otimes f) +
 \frac{1}{2} (h \otimes e - e \otimes h) +
 \frac{(2x_{3}-1)^{2}}{8 x_{1}^{2}} (h \otimes f - f \otimes h) \right).
$$
If $x_3\neq \displaystyle \frac{1}{2}$, then
taking $k=y_1$ and $a =\displaystyle  \frac{2x_{3}-1}{2x_{1}}$ in the above equation we obtain $r_3$.
If $x_{3} = \displaystyle \frac{1}{2}$, then the above equation becomes $r_1$ by taking $k=\displaystyle -\frac{1}{2}y_{1}$.

\subsection{Solutions of CYBE in $\speciall \ltimes_{{\rm ad}^{\ast}} \speciall^{\ast}$}
\label{subsec:so}

Let $\rho: \frakg \rightarrow gl(V)$ be a representation of a Lie algebra $\frakg$.
On the vector space $\frakg \oplus V$,
there is a natural Lie algebra structure
(denoted by $\frakg \ltimes_{\rho} V$ ) given by
\begin{equation}
[x_{1} + v_{1}, x_{2} + v_{2}] = [x_{1}, x_{2}] + \rho(x_{1})v_{2} -
\rho(x_{2})v_{1}, x_{1}, x_{2} \in \frakg, v_{1}, v_{2} \in V.
\end{equation}
Let $\rho^{\ast}: \frakg \rightarrow gl(V^{\ast})$
 be the dual representation of the representation
$\rho : \frakg \rightarrow gl(V)$ of the Lie algebra $\frakg$.

A linear map $P : V \rightarrow \mathfrak{g}$
can be identified as an element in $\mathfrak{g}\otimes V^{\ast}
\subset (\mathfrak{g} \ltimes_{\rho^{\ast}} V^{\ast} ) \otimes
(\mathfrak{g} \ltimes_{\rho^{\ast}} V^{\ast} )$ as follows.
Let $\{e_{1},\cdots , e_{n}\}$ be a basis of $\mathfrak{g}$.
Let $\{v_{1}, \cdots , v_{m}\}$ be a basis of $V$ and
$\{v^{\ast}_{1},\cdots, v^{\ast}_{m}\}$ be its dual basis,
that is $v^{\ast}_{i} (v_{j}) = \delta_{ij}$.
Set $P(v_{i}) =\sum\limits_{j=1}^{n} a_{ij}e_{j},$ $i=1,2,\cdots, n$.
Since as vector spaces, ${\rm Hom}(V, \mathfrak{g}) \cong \mathfrak{g}
\otimes V^{\ast}$, we have
\begin{equation}
P = \sum_{i=1}^{m} P(v_{i}) \otimes v_{i}^{\ast} =
\sum_{i=1}^{m} \sum_{j=1}^{n} a_{ij} e_{j} \otimes v_{i}^{\ast}
\in \mathfrak{g} \otimes V^{\ast} \subset
(\mathfrak{g} \ltimes_{\rho^{\ast}} V^{\ast} ) \otimes
(\mathfrak{g} \ltimes_{\rho^{\ast}} V^{\ast} ).
\end{equation}
For any tensor element $r=\sum_i a_i\otimes b_i\in V \otimes V$, denote $r^{21}=\sum_ib_i\otimes a_i$.

\begin{lemma}\mlabel{lem:op}\mcite{Bai0} Let $\frak g$ be a Lie algebra. A linear map $P: \frak g\rightarrow \frak g$ is a Rota-Baxter operator if and only if
$r = P - P^{21}$ is a skew-symmetric solution of CYBE
in $\mathfrak{g} \ltimes_{{\rm ad}^{\ast}} \frak g^{\ast}$.
\end{lemma}

{\bf Notation:} For an algebra $A$ equipped with a bilinear product,
its {\bf (formal) characteristic matrix} is defined by
\begin{equation}
M =\left(\begin{array}{ccc} \sum_{k=1}^{n}a_{11}^{k}e_{k} &\cdots&  \sum_{k=1}^{n}a_{1n}^{k}e_{k} \\
\cdots &\cdots& \cdots \\  \sum_{k=1}^{n}a_{n1}^{k}e_{k} &\cdots &  \sum_{k=1}^{n}a_{nn}^{k}e_{k}\end{array} \right),
\end{equation}
where $\{e_{1}, e_{2}, \cdots, e_{n}\}$ is a basis of $A$ and $e_{i}e_{j}= \sum_{k=1}^{n} a^{k}_{ij}e_{k}$.

Let ${\rm ad}^{\ast}$ be the coadjoint representation of $\speciall$. Then the characteristic matrix of the
6-dimensional Lie algebra $\speciall \ltimes_{{\rm ad}^{\ast}} \speciall^{\ast}$ with respect to the basis $\{e,f,h,e^{\ast},f^{\ast},h^{\ast}\}$ is
$$
\left( \begin{array}{cccccc}
 0 & h & -2e &2h^{\ast}&0&-f^{\ast}\\
 -h & 0 & 2f &0&-2h^{\ast}&e^{\ast}\\
 2e & -2f & 0&-2e^{\ast}&2f^{\ast}&0\\
-2h^{\ast}&0&2e^{\ast}&0&0&0\\
0&2h^{\ast}&-2f^{\ast}&0&0&0\\
f^{\ast}&-e^{\ast}&0&0&0&0
\end{array} \right).
$$

By Lemma \mref{lem:op}, we can obtain the following skew-symmetric solutions of CYBE in $\speciall \ltimes_{{\rm ad}^{\ast}} \speciall^{\ast}$
by the Rota-Baxter operators on $\speciall$ given in Theorem~\mref{thm:rbo}.

\threech {$r_{1} = h \otimes f^{\ast} - f^{\ast} \otimes h $,}{$r_{2}=0$,}
{$r_{3} =  f \otimes e^{\ast} - e^{\ast} \otimes f$,}

\twoch {$r_{4} = h \otimes h^{\ast} - h^{\ast} \otimes h$,}{$r_{5} = e \otimes f^{\ast} - f^{\ast} \otimes e$,}

\onechss{$r_{6} = (e + ah) \otimes f^{\ast} - f^{\ast} \otimes (e+ah)$,}

\onech {$r_{7} = (e+af) \otimes e^{\ast} + (\frac{1}{a}e + f) \otimes f^{\ast} - e^{\ast} \otimes (e+af) - f^{\ast} \otimes (\frac{1}{a}e + f)$,}

\onechs {$r_{8} = (e+\frac{a^{2}}{16} f) \otimes e^{\ast} + (\frac{16}{a^{2}} e -3f - \frac{8}{a}h) \otimes f^{\ast} +(af+2h) \otimes h^{\ast}
-e^{\ast} \otimes (e+\frac{a^{2}}{16} f)$ ${\ } \qquad \quad -f^{\ast} \otimes (\frac{16}{a^{2}} e -3f - \frac{8}{a}h) - h^{\ast} \otimes (af+2h)$,}

\twochs {$r_{9} = h \otimes e^{\ast} - e^{\ast} \otimes h$,}{$r_{10} = h \otimes e^{\ast} - 2f \otimes h^{\ast} - e^{\ast} \otimes h + 2h^{\ast} \otimes f$,}

\begin{enumerate}
\item[(11)] $r_{11} = (f+ah) \otimes e^{\ast} - e^{\ast} \otimes (f+ah)$,
\item[(12)] $r_{12} = (f+ah) \otimes e^{\ast} - 2af\otimes h^{\ast} - e^{\ast} \otimes (f+ah) + 2ah^{\ast} \otimes f$,
\item[(13)] $r_{13} = (f+ah) \otimes e^{\ast} + (2af+2a^{2}h) \otimes h^{\ast} - e^{\ast} \otimes (f+ah) - h^{\ast} \otimes (2af+2a^{2}h)$,
\item[(14)] $r_{14} = (f+ah) \otimes e^{\ast} -(4a^{2}f + 4a^{3}h) \otimes f^{\ast} +(4af+4a^{2}h) \otimes h^{\ast} -e^{\ast} \otimes (f+ah)\\
    {\ }\qquad +
f^{\ast} \otimes (4a^{2}f + 4a^{3}h) - h^{\ast} \otimes (4af+4a^{2}h)$,
\item[(15)] $r_{15} = h\otimes f^{\ast} -2e \otimes h^{\ast} - f^{\ast} \otimes h + 2h^{\ast} \otimes e$,
\item[(16)] $r_{16} = (e+ah) \otimes f^{\ast} -2ae \otimes h^{\ast} - f^{\ast} \otimes (e+ah) + 2ah^{\ast} \otimes e$,
\item[(17)] $r_{17} = (e+ah) \otimes f^{\ast} +(2ae + 2a^{2}h) \otimes h^{\ast} - f^{\ast} \otimes (e+ah) - h^{\ast} \otimes (2ae + 2a^{2}h),$
\item[(18)] $r_{18} = -(4a^{2}e  + 4a^{3} h) \otimes e^{\ast} + (e+ah) \otimes f^{\ast} + (4ae+4a^{2}h) \otimes h^{\ast} +
e^{\ast} \otimes (4a^{2}e  + 4a^{3} h)\\
{\ }\qquad -f^{\ast} \otimes (e+ah) - h^{\ast} \otimes (4ae+4a^{2}h) $,
\item[(19)] $r_{19} = (e - \frac{3a^{2}}{4}f  +ah) \otimes e^{\ast} - (\frac{4}{27a^{2}}e + \frac{1}{3}f ) \otimes f^{\ast}
- (\frac{8}{9a} e + \frac{2}{3}h) \otimes h^{\ast}   - e^{\ast} \otimes  (e - \frac{3a^{2}}{4}f  +ah)\\
{\ }\qquad + f^{\ast} \otimes (\frac{4}{27a^{2}}e + \frac{1}{3}f )
+ h^{\ast} \otimes (\frac{8}{9a} e + \frac{2}{3}h)$,
\item[(20)] $r_{20} = (ae - \frac{a^{2}}{2} h) \otimes e^{\ast} - (af + \frac{1}{2}h) \otimes f^{\ast} + (e+a^{2}f) \otimes h^{\ast}
-e^{\ast} \otimes (ae - \frac{a^{2}}{2} h) + f^{\ast} \otimes (af + \frac{1}{2}h) - h^{\ast} \otimes  (e+a^{2}f) $,
\item[(21)] $r_{21} = (ae + 4a^{3} f) \otimes e^{\ast} - (\frac{1}{4a}e +af) \otimes f^{\ast}+ (e+ 4a^{2}f) \otimes h^{\ast}
-e^{\ast} \otimes (ae + 4a^{3} f)$ \\ ${\ }\qquad +f^{\ast} \otimes (\frac{1}{4a}e +af) -h^{\ast} \otimes  (e+ 4a^{2}f)$,
\item[(22)] $r_{22} = -(\frac{1}{4b}e -af+\frac{1+16ab^{3}}{16b^2}h)\otimes e^{\ast} +(be-4ab^{2}f+\frac{1+16ab^{3}}{4}h) \otimes f^{\ast} \\
    {\ }\qquad +(e-4abf+\frac{1+16ab^{3}}{4b}h) \otimes h^{\ast} + e^{\ast} \otimes (\frac{1}{4b}e -af+\frac{1+16ab^{3}}{16b^2}h)
    \\
    {\ } \qquad -
f^{\ast} \otimes (be-4ab^{2}f+\frac{1+16ab^{3}}{4}h) - h^{\ast} \otimes (e-4abf+\frac{1+16ab^{3}}{4b}h)$.
\end{enumerate}

\section{Induced pre-Lie algebras and solutions of CYBE in 6-dimensional Lie algebras}

A {\bf pre-Lie algebra} is a vector space with a bilinear product $\{ ,\}$
satisfying that for any $x, y, z \in A$,
\begin{equation}\mlabel{eq:prelie}
\{\{x,y\},z\} - \{x, \{y,z\}\} = \{ \{y,x\},z\} - \{y,\{x,z\}\}.
\end{equation}

\begin{lemma}\mcite{Agu}\mlabel{lem:pre}
Let $(\mathfrak{g},[\cdot, \cdot])$ be a Lie algebra and $P:\frak g\rightarrow \frak g$ be a Rota-Baxter operator. Define a new operation $\{x,y\} = [P(x),y]$. Then $(\mathfrak{g},\{\cdot, \cdot\})$ is a pre-Lie algebra.
\end{lemma}

Let $(A, \{,\})$ be a pre-Lie algebra.  Then the commutator
 (cf. \mcite{Bu})
\begin{equation}
[x, y] = \{x,y\} - \{y,x\},
\end{equation}
defines a Lie algebra $\frak g(A)$, which is called the
{\bf sub-adjacent Lie algebra} of $A$ and $A$ is also called a {\bf compatible pre-Lie algebra} structure on the Lie algebra
$\frak g(A)$. Furthermore, for any $x\in A$, let $L_x$ denote the left
multiplication operator, that is, $L_x(y)=\{x,y\}$ for any $y\in A$. Then
$L:\frak g(A)\rightarrow gl({\frak g}(A))$ with $x\rightarrow L_x$
gives a representation of the Lie algebra $\frak g(A)$,
that is,
\begin{equation}
[L_x,L_y]=L_{[x,y]},\;\;\forall x,y\in A.
\mlabel{eq:rep}
\end{equation}

By Lemma ~\mref{lem:pre}, the Rota-Baxter operators on $\speciall$ obtained in Theorem ~\mref{thm:rbo} can be used to give 3-dimensional pre-Lie algebras.
Of course, different Rota-Baxter operators on $\speciall$ might give isomorphic pre-Lie algebra structures. On the other hand, all complex 3-dimensional pre-Lie algebras have been classified in \mcite{Bai}. Thus it is interesting to determine the induced 3-dimensional pre-Lie algebras in the sense of Lemma~\mref{lem:pre} up to (algebraic) isomorphisms. For this purpose we adapt the notations from~\mcite{Bai}.

\begin{theorem} \label{thm:induced} With the notations in Theorem~\ref{thm:rbo},
let $PL_{P_{i}}$ denote the pre-Lie algebra constructed by $P_{i}$
in the sense of Lemma~\ref{lem:pre}, $1\leq i \leq 22$.
Then $PL_{P_i}$ give the following pre-Lie algebras:

\begin{table}[hbtp] \begin{center}\begin{tabular}{|c|c|}
\hline
 $PL_{P_i}, 1\leq i\leq 22$ & characteristic matrices \\ \hline\hline
$PL_{P_2}$ & Trivial \\ \hline
$PL_{P_1},PL_{P_6},PL_{P_9},PL_{P_{11}}, PL_{P_{14}},PL_{P_{18}},PL_{P_{21}}, PL_{P_{22}}$ ($16ab^3\ne 1$) & $($N-1$)_{-1}$\\ \hline
$PL_{P_3}, PL_{P_5}, PL_{P_{22}}$ ($16ab^3=1$)&$($H-6$)$\\ \hline
$PL_{P_4}, PL_{P_7}, PL_{P_{13}}, PL_{P_{17}}$&$(D_{-1}$-1$)$ \\ \hline
$PL_{P_8}, PL_{P_{12}}, PL_{P_{16}}, PL_{P_{19}}$ &$($E-6$)$\\ \hline
$PL_{P_{10}}, PL_{P_{15}}, PL_{P_{20}}$ &$(\overline{D}_1$-8$)$\\ \hline
\end{tabular}
\end{center}
\end{table}
Here the entries in the right column are the following (formal) characteristic matrices of the pre-Lie algebras in~\mcite{Bai}:
\begin{eqnarray*}
&(\text{N-1})_{-1} := \left( \begin{array}{ccc}
 0 & 0 & 0 \\
 0 & 0 & 0 \\
 0 & e_{2} & - e_{3}
\end{array} \right), \quad (\text{H-6}) := \left( \begin{array}{ccc}
 0 & 0 & 0 \\
-e_{3} & e_{1} & 0 \\
 0 & 0 & 0
\end{array} \right), \quad
(D_{-1}\text{-1})_{0} := \left( \begin{array}{ccc}
0 & 0 & 0 \\
0 & 0 & 0 \\
e_{1}  & -e_{2} & 0
\end{array} \right),&\\
&(E\text{-}6) := \left( \begin{array}{ccc}
 0 & 0 & 0 \\
 0 & e_{1} & -e_{1} - e_{2} \\
  e_{1} & 0 & -e_{2}-e_{3}
\end{array} \right),\quad (\overline{D}_{1}\text{-}8) := \left( \begin{array}{ccc}
 0 & 0 & 0 \\
0 & e_{1} & -e_{2} \\
 e_{1} & 0 & -e_{3}
\end{array} \right).&
\end{eqnarray*}
\end{theorem}

\begin{proof}
It follows from direct checking that the characteristic matrix of $PL_{P_i}$ for each $1\leq i\leq 22$ is the designated matrix under the following basis $e_1:=e_{i,1}, e_2:=e_{i,2}, e_3:=e_{i,3}$:

\begin{center}\begin{tabular}{|c|c|}
\hline
 $i$ & $e_1,~ e_2,~ e_3$ \\ \hline\hline
1 & $h,~ e, ~ \frac{1}{2}f$\\ \hline
2 &  $ e,~ f,~ h$ \\ \hline
3 & $-h, ~ e, ~ 2f$\\ \hline
4 & $e, ~ f, ~ \frac{1}{2} h$ \\ \hline
5 & $h, ~ f, ~ 2e$ \\ \hline
6 & $e +ah, ~ \frac{1}{4a^{3}} e, ~ - \frac{1}{4a^{2}} (h-2af) - \frac{1}{8a^{3}}e$ \\ \hline
7 &  $e-af+\sqrt{a}h, ~ e-af-\sqrt{a}h, ~  -\frac{1}{4\sqrt{a}} (e+af)$ \\ \hline
8 & $\frac{a}{16}(f + \frac{4}{a}h - \frac{16}{a^{2}}e),
~ \frac{1}{4} h - \frac{2}{a} e, \frac{4}{a}e - \frac{1}{4}h$ \\ \hline
9 & $h, ~ f, ~  - \frac{1}{2}e$\\ \hline
10 & $ -f, ~\frac{1}{2}h, ~ -\frac{1}{2}e$ \\ \hline
11 & $f +ah, ~ -\frac{1}{4a^{3}} f, ~  \frac{1}{4a^{2}} (h-2ae) + \frac{1}{8a^{3}} f$ \\ \hline
12 & $-\frac{1}{4a^{3}} f , ~ \frac{1}{4a^{2}} h, ~ - \frac{1}{4a^{2}}h - \frac{1}{2a} e$ \\ \hline
13 & $-\frac{1}{4a^{3}} f, ~ \frac{1}{2a} e - \frac{1}{4a^{2}}h -\frac{1}{8a^{3}} f,
~-\frac{1}{4a^{2}}h -\frac{1}{4a^{3}} f$ \\ \hline
14 & $h +\frac{1}{a} f , ~ -\frac{1}{2a} e + \frac{1}{4a^{2}} h + \frac{1}{8a^{3}} f,
~ -\frac{1}{8a^{3}}f$ \\ \hline
15 & $e, ~\frac{1}{2} h, ~ \frac{1}{2}f$ \\ \hline
16 & $\frac{1}{4a^{3}}e, ~  -\frac{1}{4a^{2}}h,
~  \frac{1}{4a^{2}}h+\frac{1}{2a}f$ \\ \hline
17 & $ \frac{1}{4a^{3}}e, ~ \frac{1}{2a}f - \frac{1}{4a^{2}}h - \frac{1}{8a^{3}}e,
 ~ -\frac{1}{4a^{2}}h - \frac{1}{4a^{3}}e$ \\ \hline
18 & $ h +\frac{1}{a}e, ~
e_{2} = f-\frac{1}{2a}h- \frac{1}{4a^{2}}e,
~ e_{3} = \frac{1}{8a^{3}}e$ \\ \hline
19 & $  \frac{1}{2a}(e - \frac{9}{4}a^{2}f + \frac{3}{2}ah),
~ \frac{3}{4}(h-3af), ~ \frac{3}{4}(6af - h)$ \\ \hline
20 & $  -e+a^{2} f +ah ,
~ \frac{1}{\sqrt{2}} (h +2af), ~  -f$ \\ \hline
21 & $  e+4a^{2}f,
~  \frac{1}{4a} h + \frac{1}{2}f - \frac{1}{8a^{2}} e,
~ \frac{1}{8a} h - \frac{1}{4} f + \frac{1}{16a^{2}}e$ \\ \hline
22 ($16ab^{3} = 1$) & $ h-\frac{1}{b}f,
~ \frac{1}{\sqrt{b}}f,
~ \frac{2}{\sqrt{b}} (e + \frac{1}{2b} h - \frac{1}{4b^{2}} f)$ \\ \hline
22 ($16ab^{3} \neq 1$) & $ e -4abf + \frac{1+16ab^{3}}{4b} h,
~h-\frac{1}{b} f - \frac{4b}{16ab^{3}-1} e_{1},
~\frac{32ab^{3}}{(16ab^{3}-1)^{2}}f -
\frac{2b}{(16ab^{3}-1)^{2}} h - \frac{8b^{2}}{(16ab^{3}-1)^{3}} e_{1}$ \\ \hline
\end{tabular}
\end{center}
\end{proof}

Theorem~\mref{thm:induced} has a direct consequence on the construction of solutions of CYBE due to the following result:
\begin{prop}(\cite{Bai}) Let $A$ be a pre-Lie algebra. Then
\begin{equation}
 r = \sum_{i=1}^n (e_i\otimes e_i^* - e_i^* \otimes e_i)\end{equation}
is a solution of the classical Yang-Baxter equation in the Lie algebra ${\frak g} (A)\ltimes_{L^*}{\frak g} (A)^*$, where
$\{e_1,\cdots,e_n\}$ is a basis of $A$ and $\{e_1^*, \cdots, e_n^*\}$ is the dual basis, $L^*$ is the dual representation of the representation $L:{\frak g}\to gl({\frak g}(A))$ in Eq.~$($\mref{eq:rep}$)$.
\end{prop}

Let $A$ run through the five pre-Lie algebra structures on $\speciall$ in Theorem~\ref{thm:induced}. Then we obtain five 6-dimensional Lie algebras ${\frak g} (A)\ltimes_{L^*}{\frak g} (A)^*$. Their characteristic matrices with respect to the bases $\{e_1,e_2,e_3\}$ given in Theorem~\ref{thm:induced} and their dual bases $\{e_1^*,e_2^*, e_3^*\}$ are

$$
\left( \begin{array}{cccccc}
 0 & 0 & 0 &0&0&0\\
 0 & 0 & -e_{2} &0& 0 &0\\
 0 &e_{2} & 0&0&e_{2}^{\ast}&-e_{3}^{\ast}\\
0&0&0&0&0&0\\
0&0&-e^{\ast}_{2}&0&0&0\\
0&0&e_{3}^{\ast}&0&0&0
\end{array} \right), \quad \left( \begin{array}{cccccc}
 0 & e_{3} & 0 &0&0&0\\
 -e_{3} & 0 & 0 &e_{2}^{\ast}& 0 &-e_{1}^{\ast}\\
 0 & 0 & 0&0& 0 & 0 \\
0&-e_{2}^{\ast}&0&0&0&0\\
0&0&0&0&0&0\\
0&e_{1}^{\ast}&0&0&0&0
\end{array} \right),
$$

$$
\left( \begin{array}{cccccc}
 0 & 0 & -e_{1} &0&0&0\\
 0 & 0 & e_{2} &0& 0 &0\\
 e_{1} &-e_{2} & 0&e_{1}^{\ast}&-e_{2}^{\ast}&0\\
0&0&-e_{1}^{\ast}&0&0&0\\
0&0&e_{2}^{\ast}&0&0&0\\
0&0&0&0&0&0
\end{array} \right),
\quad
\left( \begin{array}{cccccc}
 0 & 0 & -e_{1} &0&0&0\\
 0 & 0 & -e_{1}-e_{2} &e_{2}^{\ast}-e_{3}^{\ast}& -e_{2}^{\ast}&0\\
 e_{1} & e_{1}+e_{2} & 0&e_{1}^{\ast}& -e_{2}^{\ast} & -e_{3}^{\ast} \\
0&e_{3}^{\ast}-e_{2}^{\ast}&-e_{1}^{\ast}&0&0&0\\
0&e_{2}^{\ast}&e_{2}^{\ast}&0&0&0\\
0&0&e_{3}^{\ast}&0&0&0
\end{array} \right),
$$

$$
\left( \begin{array}{cccccc}
 0 & 0 & -e_{1} &0&0&0\\
 0 & 0 & -e_{2} &e_{2}^{\ast}& -e_{3}^{\ast} &0\\
 e_{1} &e_{2} & 0&e_{1}^{\ast}&0&-e_{3}^{\ast}\\
0&-e_{2}^{\ast}&-e_{1}^{\ast}&0&0&0\\
0&e_{3}^{\ast}&0&0&0&0\\
0&0&e_{3}^{\ast}&0&0&0
\end{array} \right).
$$

Moreover, there is a solution of CYBE with the following uniform form in the above Lie algebras:
\begin{equation}\label{eq:r}
 r = \sum_{i=1}^3 (e_i\otimes e_i^* - e_i^* \otimes e_i).\end{equation}

\section{Conclusions and discussion}
Based on the study in the previous sections, we give the following conclusions and discussion.

We have given all the Rota-Baxter operators on the 3-dimensional complex simple Lie algebra $\speciall$ under the Cartan-Weyl basis. It is known that the set of Rota-Baxter operators is only a ``set" with scalar multiplication, but without further known structures. So we determine the Rota-Baxter operators up to scalar multiplication. Our result is given in our particular choice of the Cartan-Weyl basis. We would like to point out that the classification of  Rota-Baxter operators on $\speciall$ can de done directly under the orthonormal basis in place of the Cartan-Weyl basis. However, it still needs a similar complicated computational process. On the other hand, it
would be interesting to consider suitable equivalence relations of these operators such that the resulting equivalence classes are independent of the choice of bases.
It is also interesting to consider the generalization of the study in this paper to other simple complex Lie algebras. Note that for these Lie algebras, the Cartan-Weyl basis can be explicitly expressed.

We summarize the three approaches that we take in this paper to derive solutions of CYBE from Rota-Baxter operators on  $\speciall$.
\begin{enumerate}
\item[(1)] In the skew-symmetric cases (under an orthonormal basis associated to the Killing form), every Rota-Baxter operators on  $\speciall$
corresponds to a skew-symmetric solution of CYBE in  $\speciall$. As expected from the work of Semenov-Tian-Shansky, they give all the skew-symmetric solutions of CYBE in $\speciall$. They also coincide with the classification result of Belavin and Drinfeld~\cite{BD} in terms of 2-dimensional non-abelian subalgebras of $\speciall$, for which we give some details.
\item[(2)] Every Rota-Baxter operator on  $\speciall$ gives a skew-symmetric solution of CYBE in the 6-dimensional Lie algebra $\speciall \ltimes_{{\rm ad}^{\ast}} \speciall^{\ast}$.
\item[(3)] Every Rota-Baxter operator on  $\speciall$ induces a 3-dimensional pre-Lie algebra. Many of the induced pre-Lie algebras are isomorphic and there are exactly 5 non-trivial induced pre-Lie algebras. For every induced pre-Lie algebra $A$, Eq.~(\ref{eq:r}) gives a skew-symmetric solution of CYBE in the 6-dimensional Lie algebra
${\frak g} (A)\ltimes_{L^*}{\frak g} (A)^*$. We remark that these Lie algebras are not isomorphic to $\speciall \ltimes_{{\rm ad}^{\ast}}
\speciall^{\ast}$ since the latter has a $\speciall$ subalgebra.
\end{enumerate}

It is also natural to consider the Lie bialgebra structures and their quantization related to the explicit construction of (skew-symmetric) solutions of CYBE in the Lie algebras with the same explicit structure constants as in this paper. In fact, for these solutions $r$, there is a Lie bialgebra structure constructed by
\begin{equation}
\delta(x)=[x\otimes 1+1\otimes x, r].
\end{equation}
The structure constants can be expressed explicitly and it is natural to consider whether there is a corresponding Drinfel'd quantum twist~\cite{CP,D}.

We finally mention that the study of Rota-Baxter operators on Lie algebras does not follow from that of Rota-Baxter operators on associative algebras by commutator since there there exist Lie algebras (such as semisimple Lie algebras) which are not the commutators of any associative algebras.

\smallskip

\noindent {\bf Acknowledgements: } C. Bai would like to thank the support by NSFC (11271202, 11221091) and SRFDP (20120031110022).
L. Guo acknowledges support from NSF grant DMS 1001855.

\end{document}